\documentclass[
    reprint,
    superscriptaddress,
    amsmath,
    aps,
    prl,
]{revtex4-2}

\usepackage{amsmath,amsthm}
\usepackage{newpxtext,newpxmath}
\usepackage{graphicx}
\usepackage{color}
\usepackage{amsfonts}
\usepackage{braket}
\usepackage{dsfont}
\usepackage{mathtools}
\usepackage{mathrsfs}
\usepackage[cal=boondox]{mathalfa}
\usepackage{bibunits}
\usepackage{etoolbox}
\makeatletter
\AtBeginDocument{%
    \global\let\auto@bib\@empty
    \patchcmd{\endthebibliography}{\label{LastBibItem}}%
        {\label{LastBibItem-\bibunitid}}{}%
        {\PackageError{arxiv}{Cannot patch bibliography label}{}}
}
\newcommand{\setbibunitid}[1]{%
    \def\bibunitid{#1}%
    \def\@extra@b@citeb{.#1}%
    \def\@extra@binfo{.#1}%
}
\newcommand{\supplementbibnotes}{\let\@FMN@list\@empty}
\newcommand{\putunitbibliography}{%
    \immediate\write\@bibunitaux{\string\citation{REVTEX42Control,apsrev42Control}}%
    \expandafter\let\expandafter\r@LastBibItem
        \csname r@LastBibItem-\bibunitid\endcsname
    \putbib[bibliography,\jobname Notes]%
}
\makeatother
\usepackage{hyperref}

\newif\ifsupplementtoc
\let\originaladdtocontents\addtocontents
\renewcommand{\addtocontents}[2]{%
    \ifstrequal{#1}{toc}{%
        \ifsupplementtoc\originaladdtocontents{#1}{#2}\fi
    }{\originaladdtocontents{#1}{#2}}%
}

\usepackage{booktabs}
\usepackage{makecell}

\newtheorem*{definition*}{Definition}
\newtheorem{proposition}{Proposition}
\newtheorem{corollary}{Corollary}
\newtheorem{lem}{Lemma}

\NewDocumentCommand{\ketbra}{m g}{%
    \IfNoValueTF{#2}{\ket{#1}\!\bra{#1}}{\ket{#1}\!\bra{#2}}%
}

\newcommand{\eye}{\mathds{1}}
\newcommand{\Tr}{\text{Tr}}
\DeclareMathOperator{\tr}{Tr}
\DeclareMathOperator{\diag}{diag}
\newcommand{\D}{\text{d}}
\newcommand{\lind}{{\mathscr{L}}}
\newcommand{\J}{\mathbf{J}}

\usepackage{xspace}
\newcommand{\mname}{NEAT\xspace}

\begin{document}
\setbibunitid{main}
\begin{bibunit}[apsrev4-2]

\title{Cooling more while dissipating less: \\anomalous tradeoffs from thermodynamic geometry}

\author{Andrea Canzio}
\email{andrea.canzio@sns.it}
\affiliation{NEST, Scuola Normale Superiore, I-56126 Pisa, Italy}

\author{Vasco Cavina}
\affiliation{NEST, Scuola Normale Superiore, I-56126 Pisa, Italy}

\author{Salvatore Raia}
\affiliation{NEST, Scuola Normale Superiore, I-56126 Pisa, Italy}

\author{Vito Vetrano}
\affiliation{NEST, Scuola Normale Superiore, I-56126 Pisa, Italy}

\author{Matteo Scandi}
\affiliation{Instituto de Física Teórica UAM/CSIC, C. Nicolás Cabrera 13-15, Cantoblanco,
28049 Madrid, Spain}

\author{Paolo Abiuso}
\email{paolo.abiuso@oeaw.ac.at}
\affiliation{Institute for Quantum Optics and Quantum Information - IQOQI Vienna,
    Austrian Academy of Sciences, Boltzmanngasse 3, A-1090 Vienna, Austria}

\begin{abstract}
    We show that, starting from the same state, optimal driving can sometimes reach a colder (or hotter) target in the same amount of time while dissipating less energy, or equivalently dissipate the same amount of energy in a shorter time.
    We study this counterintuitive phenomenon in the slow-driving regime of an open quantum system, where we name it Near-Equilibrium Anomalous Tradeoff (NEAT).
    We establish necessary conditions for its occurrence and show its underlying geometric origin. While NEAT is reminiscent of the well-known Mpemba effect, the two phenomena are distinct in general.
    We demonstrate NEAT in an exactly solvable qutrit model, as well as for $d$-dimensional birth–death dynamics with independently controlled energy levels.
    Finally, we show how NEAT can improve the performance of Carnot engines.
    \end{abstract}

\maketitle

\begin{table}[t]
\footnotesize
\renewcommand{\arraystretch}{1.15}
\setlength{\tabcolsep}{4pt}
\begin{tabular*}{\columnwidth}{@{\extracolsep{\fill}} l c c}
\toprule
 & \textbf{NEAT} & \textbf{Mpemba} \\
\midrule
Process class & Driven processes & Free relaxations \\
Anomaly & Dissipation $\times$ Time & Time \\
Typical instance & Close to equilibrium & Far from equilibrium \\
\addlinespace[5pt]
Origin
 & \makecell{Global property of\\thermo-geometry}
 & \makecell{Local property of\\dynamical generator} \\
\addlinespace[5pt]
\makecell[l]{Improves the\\performance of} & Carnot engines & Otto engines \\
\bottomrule
\end{tabular*}
\caption{The NEAT effect uncovered in this work is reminiscent of the
well-known Mpemba effect. Nevertheless, the two phenomena have clear-cut
distinct origins and properties, which we discuss below. In fact, there
exist systems featuring NEAT without the Mpemba effect---and vice versa.}
\label{tab:neat-mpemba}
\end{table}

\textbf{\emph{Introduction.---}}
Nonequilibrium systems can exhibit counterintuitive phenomena, from dissipative phase transitions~\cite{sieberer2025universality,rodriguez2017probing} to negative differential response~\cite{baerts2013frenetic,benichou2014negative} and spontaneous synchronization~\cite{acebron2005kuramoto,lee2013synchronization}. Beyond their fundamental interest, these phenomena challenge familiar expectations about how physical processes unfold and suggest new strategies for quantum-state preparation and sensing~\cite{verstraete2009dissipation,ilias2022criticality}. Particular interest has been devoted to the concept of anomalous relaxation, including asymmetric thermalization speed~\cite{lapolla2020faster,vu2021toward,ibanez2024heating} and anomalous speedups~\cite{TEZA2026speedups,kumar2020exponentially}. A striking example of the latter is the Mpemba effect, i.e., the possibility for initially hotter water to freeze faster than initially colder water under some experimental conditions~\cite{mpemba1969cool,jeng2006mpemba}.
The term Mpemba effect has subsequently been used to describe several forms of anomalous relaxation, including the anomalous symmetry restoration associated with the quantum Mpemba effect~\cite{ares2023asymmetry,rylands2024microscopic,moroder2024thermodynamics,ares2025quantum,summer2026resource} and the anomalous hierarchy of relaxation times occurring in the Markovian Mpemba effect which, unlike its historical counterpart, does not rely on the occurrence of a phase transition~\cite{LuRaz2017nonequilibrium,klich2019mpemba,carollo2021exponentially, PhysRevLett.131.080402}.

The anomalous effects discussed above typically arise far from equilibrium, naturally raising the question of whether analogous signatures exist close to equilibrium. In this work, we answer this question in the affirmative.

We consider a system initially in thermal equilibrium with a reservoir and cool it by driving its Hamiltonian parameters while keeping the reservoir temperature fixed. In the slow-driving regime, the system remains close to the equilibrium state of its instantaneous Hamiltonian, and the dissipated work is described by thermodynamic
geometry~\cite{salamon1983length,crooks2007length,scandi2019thermodynamiclength,abiuso2020geometric}.
The associated length determines the minimum dissipative cost at fixed duration, to leading order in the slow-driving expansion~\cite{sivak2012metrics,cavina2017slowdynamics,scandi2019thermodynamiclength}.
In this context, we identify an anomalous behavior in the dissipation–time tradeoff:
increasing the temperature difference between the initial and final states can reduce the minimum dissipation required to complete the process in a fixed time.
Equivalently, a colder/hotter target can be reached in less time at the same dissipative cost.
We call this effect the \emph{Near-Equilibrium Anomalous Tradeoff} (\mname). 
While being evocative of the Mpemba effect and its path-based variants~\cite{gal2020precooling,nava2025pontus,Peluso2026optimal}, NEAT is a qualitatively different phenomenon, as we summarize in Table~\ref{tab:neat-mpemba}.

Understanding when this reversal occurs requires examining how relaxation dynamics shape the geometry of equilibrium states.
We derive conditions that constrain its occurrence and show how it emerges in a number of systems, from an analytically solvable qutrit geometry to birth–death dynamics with a structured thermal bath.
In the latter setting, an explicit expression for the dissipated work connects the anomalous temperature-dependence to the underlying transition rates.
Beyond state preparation, we discuss how this advantage can improve the power–efficiency tradeoff of
finite-time Carnot engines.

Finally, we note that, while in the following we will focus on cooling, NEAT can occur in both directions, and the entire analysis extends similarly to heating protocols.
\vspace{1em}

\textbf{\emph{The Near-Equilibrium Anomalous Tradeoff.---}}
We begin by formally defining the setup, which is illustrated in Fig.~\ref{fig:protocol}.
Consider a quantum system with Hamiltonian $H_t = \sum_a \lambda^a(t) X_a$, with $\boldsymbol{\lambda}(t)=(\lambda^a(t))_a$ the vector of
tunable control parameters, initially at thermal equilibrium with a reservoir at inverse
temperature $\beta$.
We want to cool the system to a lower temperature, corresponding to inverse temperature $\beta' > \beta$, in a time $\tau$,
while keeping it in contact with the reservoir.
The canonical equilibrium protocol achieves this by driving the Hamiltonian from $H_0$ to $H_\tau = (\beta' / \beta) H_0$,
so that at the end of the transformation the system will be in the state
\begin{align}
    \omega_{H_\tau} (\beta)
    \coloneqq \frac{e^{-\beta H_\tau}}{\tr[e^{-\beta H_\tau}]}
    = \frac{e^{-\beta' H_0}}{\tr[e^{-\beta' H_0}]}
    = \omega_{H_0} (\beta').
\end{align}

The work cost of the cooling protocol is quantified by
$W = \int_0^\tau \D t\, \tr [\rho_t \dot{H}_t]$,
with $\rho_t$ being the system state at time $t$.
In the quasi-static limit $\tau \rightarrow \infty$, the state will
be the instantaneous Gibbs state $\rho_t = \omega_{H_t}(\beta)$ at all times, 
while the work required is
$W = W_{\rm rev} \coloneqq F_{\beta}(H_\tau) - F_{\beta}(H_0)$,
where $F_\beta(H) = - \beta^{-1} \ln \tr [e^{-\beta H}]$ is the equilibrium free
energy.
We observe that in this reversible limit the work cost is independent of the specific path
$\{\boldsymbol{\lambda}(t)\}_{t}$ chosen but depends solely on the endpoint
$H_\tau = (\beta'/\beta) H_0$, and it is monotonic in the target inverse temperature $\beta'$.

In any realistic setting, though, $\tau$ will be finite and some excess amount of energy $W_{\rm diss}$ will be irreversibly spent to cool the system.
Assuming that the total duration is much longer than the typical relaxation timescales of the dynamics generated by the contact with the environment, one can perturbatively
expand the total work as
\begin{equation}
    \label{eq:diss-work}
    W_{\rm diss} \coloneqq
    W - W_{\rm rev} =
    \frac{\Sigma[\{\boldsymbol{\lambda}(t)\}_{t}]}{\tau}
    + \mathcal{O}\left(\frac{1}{\tau^2}\right),
\end{equation}
where $\Sigma[\{\boldsymbol{\lambda}(t)\}_{t}] \geq 0$ quantifies the leading-order dissipation generated during a protocol
$\{\boldsymbol{\lambda}(t)\}_{t}$. The expansion~\eqref{eq:diss-work} corresponds to the so-called
\emph{low-dissipation} or \emph{slow-driving} limit~\cite{cavina2017slowdynamics,abiuso2020geometric}\footnote{The expansion assumes a fixed smooth path and schedule uniformly stretched with $\tau$, with slow driving throughout and a relaxation gap bounded away from zero along the path.
}.

Unlike in the quasi-static case, the value of $\Sigma$ is protocol-dependent and
can be optimized.
Then, it is natural to define its achievable minimum:
\begin{equation}
    \label{eq:min-diss}
    \begin{split}
        & W_{\rm diss} \tau\geq \Sigma_{\rm min}(\beta\rightarrow\beta') + \mathcal{O}(\tau^{-1}), {\;\;\rm where} \\
        & \Sigma_{\rm min}(\beta \rightarrow \beta')
        \coloneqq \min_{\substack{\{\boldsymbol{\lambda}(t)\}_{t} \in \Lambda,\\
        \scalebox{0.6}{$
            H_0=H,\; H_\tau = (\beta'/\beta) H
            $}}
        }\;\Sigma[\{\boldsymbol{\lambda}(t)\}_{t}],
    \end{split}
\end{equation}
with $\Lambda$ representing the set of achievable control histories.
$\Sigma_{\rm min}(\beta \rightarrow \beta')$
is thus a useful characterization of the
optimal dissipation-time tradeoff of the transformation.
The key observation of the present work is that, even if $W_{\rm rev}$ is monotonic in the temperature difference, this is not necessarily the case for  $\Sigma_{\rm min}(\beta \rightarrow \beta')$. In this context, we define the \mname effect as follows:

\begin{definition*} A \emph{Near-Equilibrium Anomalous Tradeoff} (\mname effect) occurs
when, while cooling a system optimally in a finite time, one can dissipate less energy by cooling it to a
lower temperature. In the notation of Eq.~\eqref{eq:min-diss},
    \begin{equation} \label{eq:neat-definition}
        \exists\ \beta_2 > \beta_1 > \beta
        \text{ s.t. }
        \Sigma_{\rm min}(\beta\to\beta_1)
        > \Sigma_{\rm min}(\beta\to\beta_2).
    \end{equation}
\end{definition*}

\begin{figure}[t]
    \centering
    \includegraphics[width=\columnwidth]{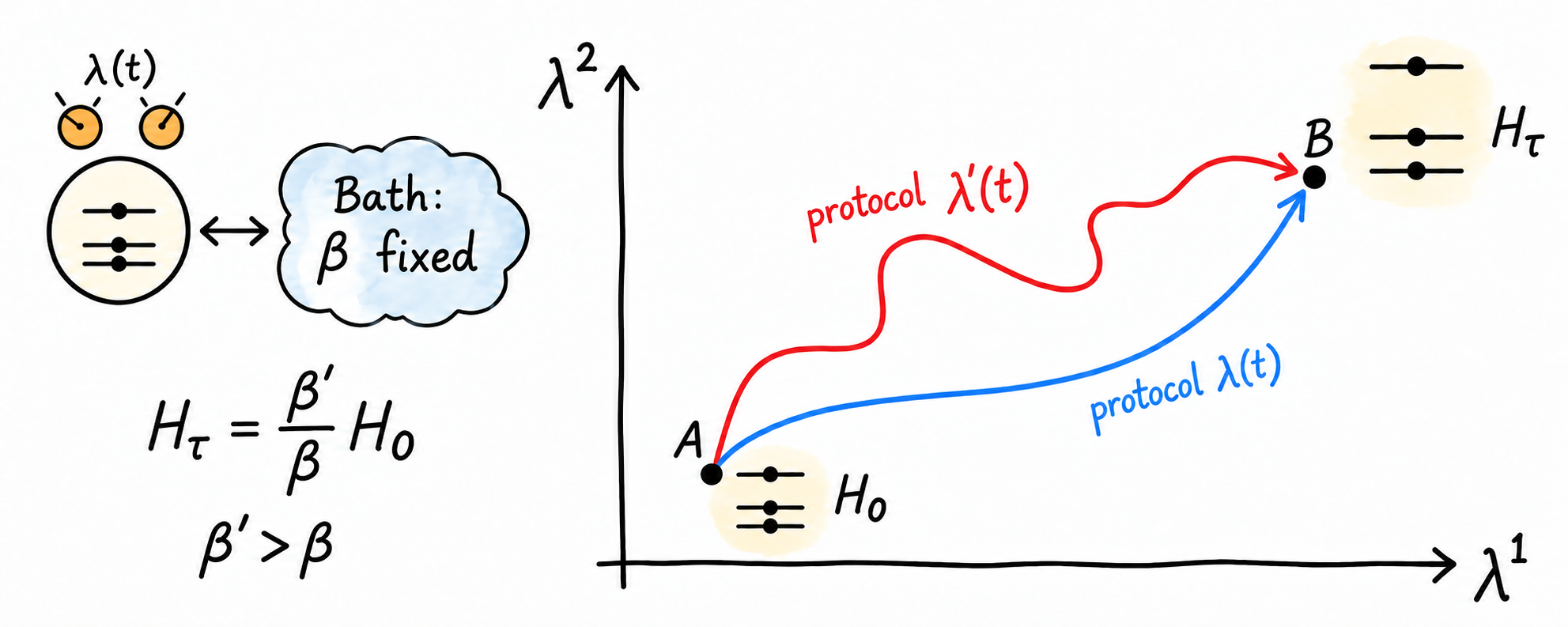}
    \caption{
        Cooling protocol at fixed bath inverse temperature $\beta$.
        Different control paths connect $H_0$ to $H_\tau=(\beta'/\beta)H_0$,
        with $\beta'>\beta$, yielding a final equilibrium state equivalent
        to cooling the initial Hamiltonian to inverse temperature $\beta'$.}
    \label{fig:protocol}
\end{figure}

In order to compute $\Sigma_{\rm min}(\beta \rightarrow \beta')$ we need to introduce a formal
description of the thermalization dynamics.
We assume that the dynamics is described by a GKSL
\emph{thermalizing} Markovian master equation~\cite{breuer2002theory}, i.e.,
any time-dependent Lindbladian whose unique instantaneous steady state is
the system Gibbs state,
\begin{align} \label{eq:scaledLind}
    \dot\rho_t = \lind_t[\rho_t],
    \qquad
    \lind_t [ \omega_{t}(\beta) ] = 0,
\end{align}
where $t\in[0,\tau]$,
and $\omega_t(\beta)\coloneqq\omega_{H_t}(\beta)$.
In the slow-driving regime, the state $\rho_t\approx\omega_t + \mathcal{O}(1/\tau)$ closely follows the instantaneous Gibbs state, with corrections that can be computed at all orders in $\tau^{-1}$~\cite{cavina2017slowdynamics}. In particular, the leading correction to work $W_{\rm diss}$~\eqref{eq:diss-work} takes the form of a quadratic form integral, defining the well-known \emph{thermodynamic length} metric, which has been derived in several contexts~\cite{salamon1983length,crooks2007length,scandi2019thermodynamiclength,abiuso2020geometric},
and measures the dissipative cost of moving along different control directions,
incorporating the near-equilibrium response induced by $\lind_t$~\cite{scandi2019thermodynamiclength} (we also provide a derivation in SM~\cite{supplemental}, Sec.~I):
\begin{align}
    \Sigma [ \{\boldsymbol{\lambda}(t)\}_{t}]
    = \tau \int_0^\tau \D t \, m_{ab}(\boldsymbol{\lambda}) \dot{\lambda}^a(t) \dot{\lambda}^b(t),
    \label{eq:quadratic-param}
\end{align}
with metric $m_{ab} \coloneqq (\tilde{m}_{ab} + \tilde{m}_{ba})/2$,
and
\begin{align}
    \label{eq:metric}
    \tilde{m}_{ab} \coloneqq - \beta \Tr \left[
        X_a \lind^+_t [ \J_{\omega_t} [X_b] ]
    \right].
\end{align}
Here, $\lind_t^+$ is the Drazin inverse of the Lindbladian~\footnote{For a Lindbladian with unique stationary state $\omega_t$, $\lind_t^+$ is its inverse on traceless operators, extended by $\lind_t^+[\omega_t]=0$. Thus $\lind_t\lind_t^+[X]=\lind_t^+\lind_t[X]=X-\omega_t\Tr X$; see Ref.~\cite{scandi2019thermodynamiclength}, Eq.~(11) and Appendix~C.}, while $\J_{\omega_t}$ is the centred Kubo--Mori map~\cite{petz1993bogoliubov,scandi2025quantum}, defined as $\J_\omega[X]\coloneqq\int_0^1\D u\,\omega^{1-u}
\bigl(X-\Tr[\omega X]\eye\bigr)\omega^u$.

Equation~\eqref{eq:quadratic-param} implies that the dissipation only depends on the geometry
induced by the Riemannian metric $m_{ab}$
and satisfies the thermodynamic length inequality
\begin{align}
    \Sigma [\{ \boldsymbol{\lambda}(t) \}_t]
    \geq
    \ell ^2 [\{ \boldsymbol{\lambda}(t) \}_t]
    \geq L^2(\beta \rightarrow \beta'),
\end{align}
where $\ell [ \{\boldsymbol{\lambda}(t)\}_{t}] \coloneqq \int_0^\tau \D t \, \sqrt{m_{ab}(\boldsymbol{\lambda}(t)) \dot{\lambda}^a(t) \dot{\lambda}^b(t) }$
is the path length
and $L(\beta \rightarrow \beta')$ is the \emph{thermodynamic length} given by
the associated geodesics:
\begin{align}
    \label{eq:thermo-length}
    L(\beta \rightarrow \beta')
    \coloneqq \inf_{\substack{\{\boldsymbol{\lambda}(t)\}_{t},\\
    \scalebox{0.6}{$
        H_0=H,\; H_\tau = (\beta'/\beta) H
        $}}
    }\; \ell [\{\boldsymbol{\lambda}(t)\}_{t}].
\end{align}
We therefore see how, in the low-dissipation regime, the \mname effect takes a completely
geometric form: the minimizer $\{ \boldsymbol{\lambda}_{\rm opt}(t) \}_t$ in Eq.~\eqref{eq:thermo-length} is by
definition the \emph{geodesic trajectory}.
If such a geodesic is achievable, i.e.~$\{ \boldsymbol{\lambda}_{\rm opt}(t) \}_t \in \Lambda$,
the thermodynamic length inequality can be saturated as
$\Sigma_{\rm min}(\beta\rightarrow\beta')
= \Sigma[\{\boldsymbol{\lambda}_{\rm opt}(t)\}_t]
= L^2(\beta \rightarrow \beta')$.
This allows us to recast the NEAT effect and study the minimally dissipative trajectories
in a geometric language, that is:

\begin{proposition}
    \label{prop:geom_effect}
    The \mname effect occurs if and only if there exist $\beta_2>\beta_1>\beta$ such that
    \begin{align}
        L(\beta\rightarrow\beta_1)
        &>L(\beta\rightarrow\beta_2).
    \end{align}
\end{proposition}

\textbf{\emph{Conditions for the \mname effect.---}}
We now provide necessary and sufficient conditions for the NEAT effect
to arise based on the latter formulation.

The first condition to impose is some degree of anisotropy in the
thermalization dynamics:
\begin{proposition}[No NEAT with uniform thermalization]
    \label{prop:single-timescale}
    No \mname effect can appear if the dynamics has a single timescale, that is, if
    the Lindbladian has the form
    \begin{align}
    \label{eq:simple_lind}
        \lind_t[\rho]=\tau_{\rm eq}^{-1} (\omega_t-\rho).
    \end{align}
\end{proposition}
The proof of  Prop.~\ref{prop:single-timescale} follows from the simple expression that the thermodynamic length takes for commuting endpoints in this case, that is (cf. End Matter):
		\begin{align}
			L(\beta,\beta')  = 2\sqrt{\frac{\tau_{\rm eq}}{{\beta}}}\arccos\Tr[\sqrt{\omega_{H}(\beta) \omega_{H}(\beta')}].
		\end{align}
Then, it is sufficient to show that the fidelity $\Tr[\sqrt{\omega_{H}(\beta) \omega_{H}(\beta')}]$ monotonically decreases as the temperature difference $|\beta-\beta'|$ increases. This can be easily checked by taking $\beta'=\lambda\beta$ and evaluating the sign of $\partial_\lambda\Tr[\sqrt{\omega_{H}(\beta) \omega_{H}(\lambda\beta)}]$, which is negative for $\lambda > 1$ and positive for $\lambda < 1$ (see End Matter for details).

Prop.~\ref{prop:single-timescale} thus shows that NEAT relies on the interplay between the geometry of the manifold of thermal states and the nontrivial structure of the dynamical generators. Interestingly, nontrivial dissipator structures are also necessary for the Markovian Mpemba effect \cite{LuRaz2017nonequilibrium, avitan2026necessary}, which occurs in a different setting (see Table \ref{tab:neat-mpemba}). Moreover, achieving NEAT requires at least two control parameters $\lambda^a(t)$.
Otherwise, even if the thermalization time scales are not uniform, the only path compatible with the final condition is given by the Hamiltonian $H_t = \lambda(t) H_0$ and its cost is monotonic in the final temperature.
\begin{proposition}[Minimum number of controls]
    \label{thm:easy-no-go}
    Let $n$ be the number of independent control parameters in the Hamiltonian, i.e.,
    $H_t = \sum_{a=1}^n \lambda^a(t) X_a$ with $\{ X_a \}_a$ linearly independent.
    The NEAT effect is possible only if $n \geq 2$.
\end{proposition}

Proposition~\ref{thm:easy-no-go} also sets constraints on the system dimension.
The number of controls $n$ is bounded above by the maximum number of independent Hermitian operators,
$n \leq d^2-1$ in general and $n \leq d-1$ for classical systems (i.e., incoherent dynamics).
These ideas are formalized in the following corollary.

\begin{corollary}[Minimum system dimension]
    \label{cor:qubitsNoGo}
    Classical systems featuring the \mname effect have dimension at least $d\geq 3$ (no \mname
    effect for bits).
    The same holds for quantum systems, under the additional assumption of  detailed-balance
    dynamics and uniqueness of the thermodynamic length geodesic (no \mname
    effect for qubits, see~\cite{supplemental}).
\end{corollary}
Notice that Prop.~\ref{thm:easy-no-go} and Cor.~\ref{cor:qubitsNoGo} are optimal: indeed, we present below an explicit example of the NEAT effect for $d-1 = n = 2$.

We now characterize the effect locally using the thermodynamic metric
$m_{ab}(\boldsymbol{\lambda})$, assumed smooth and positive definite.
By the mean value theorem, a distance inversion occurs exactly when
$\partial_{\beta'}L(\beta\to\beta')<0$ at an interior cooling target. This gives the following geometric criterion (see SM~\cite{supplemental} for details).

\begin{proposition}[Geometric condition for NEAT]
\label{prop:neat-geometric}
Assume that the minimizing geodesic $\varsigma$ from the initial state to every intermediate target is unique. Then \mname occurs if and only if, at some intermediate target $\beta_*$, the scalar product $(\cdot,\cdot)_m$ in metric $m$ between the tangent $\mathbf{v}_\varsigma(\beta_*)$ of the geodesic ending at $\omega(\beta_*)$ and the tangent $\mathbf{v}_\omega(\beta_*)$ of the Gibbs curve is negative, i.e.
\begin{equation}
    \exists\ \beta_*\ \text{ s.t. }\
    \big( \mathbf{v}_\varsigma(\beta_*),\mathbf{v}_\omega(\beta_*) \big)_{m} < 0.
    \label{eq:neat-endpoint-pairing}
\end{equation}
\end{proposition}
Thus, cooling the system further reduces the optimal cost when the two tangent vectors make an obtuse angle, as illustrated in Fig.~\ref{fig:qutrit-combined}(a) for a qutrit.

\begin{figure*}[t]
    \centering
    \begin{minipage}[t]{0.325\textwidth}
        \vspace{0pt}
        \makebox[\linewidth][l]{\textbf{(a)}}\\[2pt]
        \begin{minipage}[c][0.95\linewidth][c]{\linewidth}
            \centering
            \includegraphics[width=\linewidth]{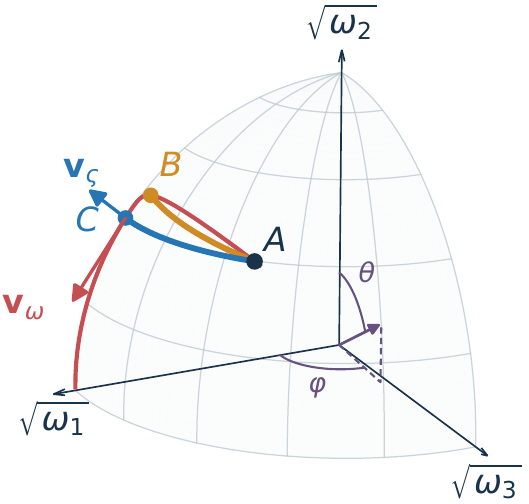}
        \end{minipage}
    \end{minipage}\hfill
    \begin{minipage}[t]{0.325\textwidth}
        \vspace{0pt}
        \makebox[\linewidth][l]{\textbf{(b)}}\\[2pt]
        \begin{minipage}[c][0.95\linewidth][c]{\linewidth}
            \centering
            \includegraphics[width=\linewidth]{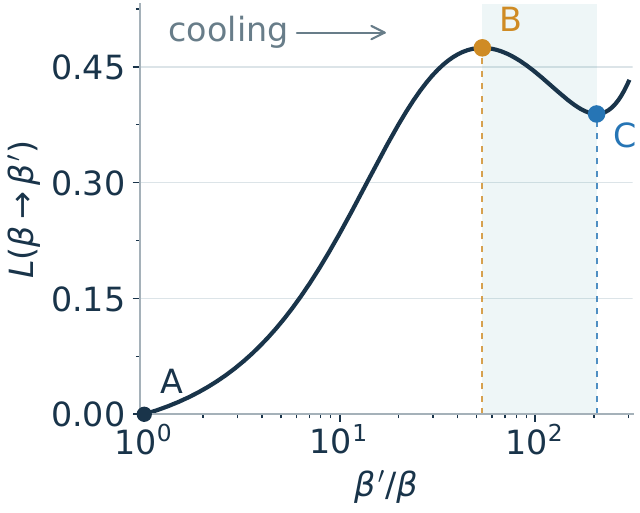}
        \end{minipage}
    \end{minipage}\hfill
    \begin{minipage}[t]{0.325\textwidth}
        \vspace{0pt}
        \makebox[\linewidth][l]{\textbf{(c)}}\\[2pt]
        \begin{minipage}[c][0.95\linewidth][c]{\linewidth}
            \centering
            \includegraphics[width=\linewidth]{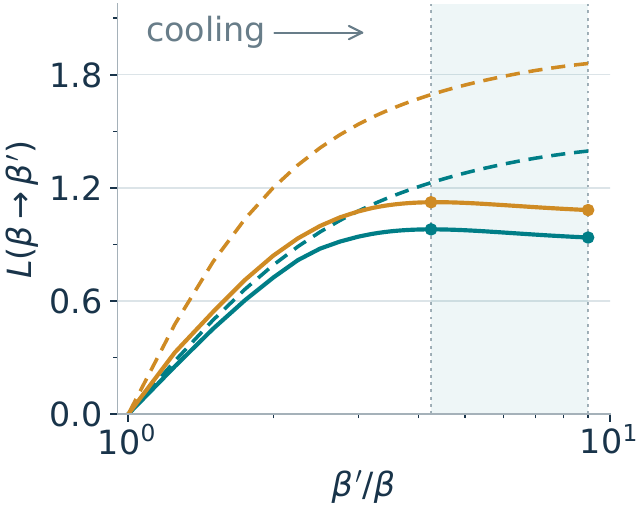}
        \end{minipage}
    \end{minipage}
    \caption{{\bf \mname in incoherent systems}.
    (a,b)  Qutrit with $H_0=\diag(0,0.03,1)$, $\beta=0.1$, $r_\theta=5$, $r_\varphi=100$.
    (a) Gibbs cooling curve (red) and optimal paths from $A$ to $B$ (amber) and $C$ (blue).
    Arrows show $\mathbf{v}_\varsigma$ and $\mathbf{v}_\omega$ at $C$, rescaled for visibility.
    The purple vector illustrates the spherical coordinates.
    (b) Exact distance along the Gibbs curve, Eq.~\eqref{eq:qutrit-endpoint-distance}; shading marks its decrease between $B$ and $C$.
    (c) Birth--death cooling for $d=3$ (light-blue) and $d=4$ (amber): structured bath (solid)
    and Ohmic control (dashed). Dots mark targets $\beta'/\beta=4.25$ and $\beta'/\beta=9$.
    In (b,c), cooling proceeds right.
    For (c), $\hbar=k_{\rm B}=\beta=1$,
    $H_3=\diag(0,0.72,2.40)$, $(c_1,c_2)=(1,0.30)$;
    $H_4=\diag(0,0.72,2.40,2.70)$, $(c_1,c_2,c_3)=(1,0.30,1)$;
    the bath in Eq.~\eqref{eq:bd-main-spectrum} has $C=10$, $\nu_0=1$,
    $\gamma_{\rm B}=0.30$, $\nu_c=10$ ($C=0$ for the control).}
    \label{fig:qutrit-combined}
    \label{fig:bd-length-curves}
\end{figure*}

In the following,
we characterize the metric $m$ (Eq.~\eqref{eq:metric}) for a wide class of incoherent processes by considering protocols in which we have control over the Hamiltonian eigenvalues but not its eigenbasis, i.e. $H_t = \sum_{a=1}^{d} e_a(t) \ketbra{a}{a}$,
$d$ being the system dimension.
Due to this constraint, the Hamiltonian dynamics cannot generate
coherence in the system. Assuming that the
thermalization process satisfies detailed balance, it follows that starting from a diagonal
(Gibbs) state the dynamics will stay diagonal at all times.
This in turn allows us to describe the system with a $d$-dimensional population vector $\boldsymbol{\rho} = (p_1,\dots,p_d)^{\mathsf T}$ (instead of a density matrix $\rho$), and the corresponding
instantaneous Gibbs populations $\boldsymbol\omega_t$.

\vspace{1cm}

\textbf{\emph{Minimal NEAT effect ($d=3$).---}}
We first study the NEAT effect in an incoherent qutrit
with Hamiltonian $H_t = \diag (0, e_2(t), e_3(t))$.
We write the Gibbs population vector $\boldsymbol{\omega}=(\omega_1,\omega_2,\omega_3)^{\mathsf T}$
in spherical coordinates $(\theta,\varphi)$, with 
$\boldsymbol{\omega} = (\sin^2 \theta \cos^2 \varphi,\cos^2 \theta, \sin^2 \theta \sin^2 \varphi)$, as shown in Fig.~\ref{fig:qutrit-combined}(a).
We now consider a thermalization process in which the Lindbladian is such that, under these coordinates, the metric is diagonal and takes the form
\begin{align}
    m_{\theta\theta} = \frac{4}{\beta r_\theta(\theta,\varphi)},
    \quad
    m_{\varphi\varphi} = \frac{4\sin^2\theta}{\beta r_\varphi(\theta,\varphi)},
    \quad m_{\theta\varphi}=0;
\end{align}
with $r_\theta, r_\varphi$ relaxation rates
in the two orthogonal directions.
For constant rates $r_\varphi\geq r_\theta>0$, the metric is such that the manifold is isomorphic to
a spherical wedge, and its geodesics
can be evaluated
exactly (see SM~\cite{supplemental}, Sec.~IV B). For any initial and final Gibbs population vectors $\boldsymbol{\omega}_A, \boldsymbol{\omega}_B$
we have
\begin{equation}
    \label{eq:qutrit-endpoint-distance}
    \begin{aligned}
    L(\boldsymbol{\omega}_A,\boldsymbol{\omega}_B)
    =&\frac2{\sqrt{\beta r_\theta}}\arccos\!\Big[\cos\theta_A\cos\theta_B
    \\
    &+\sin\theta_A\sin\theta_B
    \cos\!\Big(\sqrt{\frac{r_\theta}{r_\varphi}}
    (\varphi_B-\varphi_A)\Big)\Big].
    \end{aligned}
\end{equation}
For suitable choices of the initial temperature and energies,
this model displays the NEAT effect.
Specifically, for some target final temperatures, cooling the
system as much as ten times more can lead to less dissipated
energy, as shown in Fig.~\ref{fig:qutrit-combined}(b) (for the plotted
parameters, $\beta'/\beta = 25$ and $\beta'/\beta = 250$ give
approximately equal lengths).
This happens because for $r_\theta/r_\varphi\ll1$, changes in $\varphi$ contribute little to the thermodynamic length, allowing a colder target to be reached with less dissipation.

\textbf{\emph{NEAT in birth--death dynamics with bosonic baths.---}} The latter qutrit example shows NEAT for a physically valid Lindbladian, which, however, lacks a microscopic description.
We now consider weak-coupling thermalization described by Davies generators
at inverse bath temperature $\beta$~\cite{davies1974markovian,alicki1976detailed},
with transitions only between adjacent energy levels and rates satisfying thermal
detailed balance. The metric can be evaluated explicitly, and the associated
geodesic problem solved numerically. Writing $A_a:=\ket{a}\!\bra{a+1}$,
the dissipative part of the generator
$\lind_t[\rho] = -i[H_t,\rho] + \lind_t^{\rm diss}[\rho]$ is
\begin{align}
    \lind_t^{\rm diss}[\rho]
    =\sum_{a=1}^{d-1}\left(
    \gamma_a^\downarrow\mathcal{D}_{A_a}[\rho]
    +\gamma_a^\uparrow\mathcal{D}_{A_a^\dagger}[\rho]
    \right),
    \label{eq:bd-main-ladder}
\end{align}
where $\mathcal D_A[\rho]:=A\rho A^\dagger-\tfrac12\{A^\dagger A,\rho\}$,
and $d$ is the system dimension.
On diagonal states this is a reversible birth--death chain~\cite{crooks2018drazin}.

Denote the instantaneous Gibbs populations of the system
and their cumulative populations by
\begin{align}
    \omega_a(t) \coloneqq \frac{e^{-\beta e_a(t)}}{\sum_{b=1}^{d} e^{-\beta e_b(t)}},
    \qquad
    \Omega_a(t) \coloneqq \sum_{b=1}^{a} \omega_b(t).
    \label{eq:bd-main-cumulants}
\end{align}
Setting $e_1(t)=0$ without loss of generality,
in the cumulative coordinates $\boldsymbol\Omega=(\Omega_1,\ldots,\Omega_{d-1})^{\mathsf T}$,
the metric is diagonal (see SM~\cite{supplemental}, Sec.~IV C for the full derivation):
\begin{align}
    m_{ab}(\boldsymbol\Omega)
    =\frac{\delta_{ab}}{\beta w_a(\boldsymbol\Omega)},
    \qquad
    w_a(\boldsymbol{\Omega})
    = \gamma_a^\downarrow \omega_{a+1}
    = \gamma_a^\uparrow \omega_{a},
    \label{eq:bd-main-cdf-metric}
\end{align}
where $a,b=1,\ldots,d-1$, $\Omega_d=1$, and $\delta_{ab}$ is the Kronecker delta.
For a bosonic bath with adjacent-transition matrix elements $c_a$, spectral density
$\mathcal{J}(\nu)$, and occupation $n_\beta(\nu)=(e^{\beta\nu}-1)^{-1}$, one has transition
rates (with $\hbar=1$)
\begin{equation}
    \begin{split}
        &\gamma_a^\downarrow(t)
        = |c_a|^2 \gamma_{\rm bath}^\downarrow\left(e_{a+1}(t)-e_a(t)\right), \\
        &\gamma_{\rm bath}^\downarrow (\nu)
        = \mathcal{J}(\nu)\big[n_\beta(\nu)+1\big],
        \label{eq:bd-main-rates-down}
    \end{split}
\end{equation}
while $\gamma_a^\uparrow$ follows from the detailed balance condition.
We use an exponentially cut off Ohmic background plus a smooth underdamped Brownian mode~\cite{jha2024unraveling},
\begin{align}
    \mathcal{J}(\nu)
    ={}&\nu e^{-\nu/\nu_c}
    +C e^{-\nu_0/\nu_c}
    \frac{\gamma_{\rm B}^{2}\nu_0^{2}\nu}
    {(\nu_0^2-\nu^2)^2+\gamma_{\rm B}^{2}\nu^2},
    \label{eq:bd-main-spectrum}
\end{align}
where $\gamma_B$, $\nu_0$, $\nu_c$ and $C$ are parameters of the model.
We then evaluate the thermodynamic length numerically as a function of the
final temperature for three- and four-level systems.
Figure~\ref{fig:bd-length-curves}(c) shows a decrease
in the thermodynamic length and hence the NEAT effect near
$\beta'/\beta\simeq4.35$ for both dimensions.
Removing the damped mode ($C=0$), with all other parameters unchanged,
yields numerically monotone control curves for the purely Ohmic case on the same interval,
suggesting that anisotropy in the bath is necessary for this class
of master equations.

\textbf{\emph{Consequences for thermal machines.---}}
Anomalous Mpemba relaxation can enhance Otto-engine performance~\cite{lin2022power,liu2024speeding}.
NEAT offers a complementary advantage on the driven isothermal strokes
of a Carnot engine: it can increase the entropy swing $\Delta S$ while
reducing the thermodynamic length that governs dissipation. As an example,
for fixed bath temperatures $T_h>T_c$ and baths with matched power-law  spectral densities $\mathcal{J}(\nu)\propto \nu^{\alpha}$, the optimized low-dissipation power takes the form~\cite{abiuso2020geometric}
\begin{align}
\label{eq:Pmax}
    P_{\rm max}=\frac{\Delta S^2}{L^2_h}\frac{(T_h-T_c)^2}{4(1+(T_c/T_h)^\frac{1-\alpha}{2})^2},
\end{align}
where $L_h$ is the hot-isotherm length normalized by entropy
production. The isothermal protocol thus enters through
$(\Delta S/L_h)^2$: a NEAT effect improves both numerator
and denominator, allowing a larger thermodynamic transformation at lower
dissipative cost.
This is immediately reflected in a larger $P_{\rm max}$ in~\eqref{eq:Pmax}, and, more generally, in improvements to the several figures of merit that increase with $|\Delta S|/L$~\cite{abiuso2020geometric,abiuso2020optimal}.
An explicit Carnot cycle construction and
further details can be found in SM~\cite{supplemental}.

\textbf{\emph{Discussion.---}} In this work we uncovered a new feature of systems with non-uniform thermalization dynamics, that is a (near-equilibrium) anomalous tradeoff (NEAT): the minimum attainable value of $W_{\rm diss}\tau$ can decrease while reaching states that are further in temperature from the initial state.
That is, a colder (or hotter) target can be reached with smaller dissipation in the same time, or in a shorter time with the same dissipation.

In the slow-driving regime, the NEAT effect reflects a global property of the thermodynamic geometry induced by the relaxation dynamics over the accessible control manifold, namely, the fact that the geodesic thermodynamic length is not always monotonic in the temperature difference.
Control freedom and relaxation dynamics must work together to realize this anomaly: sufficient controls allow driven paths to explore directions with different dissipation costs. The structured-bath examples give this mechanism a concrete physical realization. Spectral resonances enhance selected transitions, making motion through parts of the control manifold less dissipative and enabling a decrease in the optimal cooling cost.

By showing that colder targets need not carry a greater
dissipative cost, NEAT is (locally) at odds with the intuition offered by the third law of thermodynamics on cooling~\cite{masanes2017general,taranto2023landauer}, opens a route to improved thermal machines
and motivates a reassessment of fidelity--dissipation--time
tradeoffs in information erasure~\cite{proesmans2020finitetime,scandi2022minimally}.

\par\vspace{\baselineskip}
\emph{Acknowledgements.---}
P.A. acknowledges funding from the Austrian Science Fund (FWF) projects 10.55776/I6004 and 10.55776/ESP2889224. M.S. acknowledges support from the Spanish Agencia Estatal de Investigación through the grants “PID2023-150847NA-I00” and “PCI2024-153448”, funded by MICIU/AEI/10.13039/501100011033 and co-funded by the EU.
V.C. acknowledges
financial support from the Italian Ministry of University and Research, through the PNRR MUR Project No.
PE0000023-NQSTI.
This publication has also been funded within the framework of the R\&D\&I Project CEX2025-001574-S, funded by MI-CIU/AEI/10.13039/501100011033. The research presented in this publication falls within the research line “Quantum Matter and Fields”. This project was funded within the QuantERA II Programme, which received funding from the European Union's Horizon 2020 research and innovation programme under Grant Agreement No.~101017733.

\par\vspace{\baselineskip}
\emph{Data availability.---}
The code and numerical data supporting this work are available in the
project repository at \url{https://github.com/TheQuantumJoker/neat-cooling}.

\putunitbibliography

\clearpage
\onecolumngrid
\section*{End Matter}

\subsection*{Proof of Proposition~\ref{prop:single-timescale}}
\begin{proof}
    For the simple Lindbladian~\eqref{eq:simple_lind}, its Drazin inverse is readily found to be
    \begin{equation}
        \lind_t[X]=\tau_{\rm eq}^{-1}(\omega_t\Tr X-X),\qquad
        \lind_t^+[X]=-\tau_{\rm eq}(X-\omega_t\Tr X).
    \end{equation}
    The resulting thermodynamic length geometry~\eqref{eq:metric} simplifies as
    \begin{align}
        \dot{W}_{\rm diss}
        = \beta\tau_{\rm eq}\Tr \left[\dot{H}_t[ \J_{\omega_t} [\dot{H}_t] ]\right]
        =\frac{\tau_{\rm eq}}{\beta}\Tr \left[\dot{\omega}_t[ \J^{-1}_{\omega_t} [\dot{\omega}_t] ]\right].
    \end{align}
    The latter expression is, up to the $\frac{\tau_{\rm eq}}{\beta}$ factor, the well-known Kubo-Mori-Bogoliubov quantum Fisher information quadratic form~\cite{scandi2025quantum,scandi2019thermodynamiclength}.
    While the geodesics associated with this metric are not known in full generality, they are known to reduce to the classical Fisher information geodesics for commuting endpoints~\cite{scandi2025quantum}, as is the case for our temperature-varying protocols $[\omega_H(\beta),\omega_H(\beta')]=0$. That is, without loss of generality, one can in this case restrict to fully incoherent protocols acting on the populations $\omega_a(t)$ defined in Eq.~\eqref{eq:bd-main-cumulants}:
    \begin{equation}
        W_{\rm diss} = \frac{\tau_{\rm eq}}{\beta} \sum_a\int_{0}^\tau \D t\,\frac{\dot{\omega}_a(t)^2}{\omega_a(t)} = \frac{4\tau_{\rm eq}}{\beta} \sum_a\int_{0}^\tau \D t\,\dot{q}_a(t)^2,\label{eq:11}
    \end{equation}
    where $q_a(t):=\sqrt{\omega_a(t)}$ and $\sum_a q_a(t)^2=1$. This is immediately recognizable as the squared Euclidean distance in the space of square roots. The geodesics are great circles on the positive spherical orthant (i.e., a $d$-dimensional octant)
    of radius $2\sqrt{\tau_{\rm eq}/\beta}$. For the fixed endpoint Hamiltonian $H=\sum_a e_a\ketbra{a}$, write $\omega_a(b):=e^{-b e_a}/\sum_c e^{-b e_c}$ for its Gibbs populations at inverse temperature $b$. The thermodynamic length is then
    \begin{equation}
        L^2(\beta\to\beta')=\frac{4\tau_{\rm eq}}\beta\left(\arccos\Tr[\sqrt{\omega_{H}(\beta)\,\omega_{H}(\beta')}]\right)^2 = \frac{4\tau_{\rm eq}}\beta\left(\arccos\sum_a\sqrt{\omega_a(\beta)\omega_a(\beta')}\right)^2.
    \end{equation}
    Then, in order to prove Proposition~\ref{prop:single-timescale}, it is sufficient to show that $\Tr[\sqrt{\omega_{H}(\beta)\,\omega_{H}(\beta')}]$ monotonically decreases as the temperature difference $|\beta-\beta'|$ increases. To simplify the notation, we set $X:=\beta H$ and define the Gibbs family and its overlap with the initial state as
    \begin{align}
        \omega(\lambda):=\frac{e^{-\lambda X}}{\Tr e^{-\lambda X}}=\omega_H(\lambda\beta),
        \qquad f(\lambda):=\Tr\!\left[\sqrt{\omega(1)\omega(\lambda)}\right].
    \end{align}
    The target corresponds to $\lambda=\beta'/\beta$. Without loss of generality, we restrict to cooling, $\lambda\geq1$. Since $[\omega(1),\omega(\lambda)]=0$, the derivative of the overlap takes the simple form
    \begin{align}
        f'(\lambda)=\frac{f(\lambda)}{2}\left(\Tr[X\omega(\lambda)]-\frac{\Tr[X\sqrt{\omega(1)\omega(\lambda)}]}{f(\lambda)}\right).
        \label{eq:derOverlap}
    \end{align}
    Setting $\bar\lambda:=(1+\lambda)/2$, we have $\sqrt{\omega(1)\omega(\lambda)}\propto e^{-\bar\lambda X}\propto\omega(\bar\lambda)$. Thus, Eq.~\eqref{eq:derOverlap} has the same sign as the difference in average energy between $\omega(\lambda)$ and $\omega(\bar\lambda)$. By the positivity of heat capacity, $f'(\lambda)<0$ for $\lambda>1$, and conversely $f'(\lambda)>0$ for $\lambda<1$. This proves that $\Sigma_{\rm min}(\beta\to\beta')$ is monotonic in the temperature difference.
\end{proof}

\subsection*{Proof of Proposition~\ref{thm:easy-no-go}}
\begin{proof}
    Suppose that $n=1$. Since $H(0)= H$ and $H(\tau) = (\beta'/\beta) H$,  without loss of generality the driving takes the form $H(t) = \lambda(t) H$, with $\lambda(0) = 1$. In this context, as the manifold is one-dimensional, the geodesic length takes the form
    \begin{align}
        L(\beta\rightarrow\beta') =\int_{1}^{\beta' / \beta} \D \lambda \;\sqrt{ m(\lambda)},
    \end{align}
	where we introduced $m(\lambda(s)) := \Tr\left[H(s)\,\J^{-1}_{\omega_s}[(-\beta^{-1}\lind_{s}^+)[H(s)]]\right]$.
	This is a restatement of a well-known fact: in one-dimensional manifolds $L(\beta\rightarrow\beta')$ does not depend on the specific path chosen but only on the endpoint $\beta'/\beta$, and therefore it is clearly monotonic in $\beta'$ (being the integral of a positive function).

\end{proof}

\end{bibunit}

\supplementtoctrue
\setbibunitid{supp}
\begin{bibunit}[apsrev4-2]
\supplementbibnotes
    \clearpage
    \onecolumngrid
    \setcounter{equation}{0}
    \setcounter{figure}{0}
    \setcounter{table}{0}
    \setcounter{section}{0}
    \setcounter{subsection}{0}
    \setcounter{proposition}{0}
    \setcounter{corollary}{0}
    \setcounter{lem}{0}
    \setcounter{secnumdepth}{2}
    \renewcommand{\theequation}{S\arabic{equation}}
    \renewcommand{\thefigure}{S\arabic{figure}}
    \renewcommand{\thetable}{S\arabic{table}}
    \renewcommand{\theHequation}{supp.\arabic{equation}}
    \renewcommand{\theHfigure}{supp.\arabic{figure}}
    \renewcommand{\theHtable}{supp.\arabic{table}}
    \renewcommand{\theHsection}{supp.\arabic{section}}
    \renewcommand{\theHsubsection}{\theHsection.\arabic{subsection}}
    \renewcommand{\theHproposition}{supp.\arabic{proposition}}
    \renewcommand{\theHcorollary}{supp.\arabic{corollary}}
    \renewcommand{\theHlem}{supp.\arabic{lem}}
    \begin{center}
        {\large\bfseries Supplemental Material}\\[0.5em]
        {\large\bfseries Cooling more while dissipating less: 
        anomalous tradeoffs from thermodynamic geometry}
    \end{center}

\tableofcontents
\clearpage

\section{Dissipated work in the slow-driving limit and thermodynamic length}
\label{sm:app:slow-driving}

We use the notation from the main text: $t\in[0,\tau]$ is physical time,
$\beta$ is the fixed bath inverse temperature, and $\beta'$ labels the target.
Dots denote time derivatives. Consider a smooth control profile traversed in time $\tau$, with a nonvanishing relaxation gap and no fast segments.
In the slow-driving limit (i.e., when $\tau$ is much larger than any thermalization timescale of the system), we can approximate the state at each time as~\cite{cavina2017slowdynamics}:
\begin{equation}
\rho_t=\omega_t+\delta\rho_t+O(\tau^{-2}),\qquad
\lind_t[\delta\rho_t]=\dot\omega_t,\qquad \Tr\delta\rho_t=0.
\end{equation}
Here $\rho_t$ is the actual state,
$\omega_t=e^{-\beta H_t}/\Tr e^{-\beta H_t}$ its instantaneous Gibbs state,
and $\delta\rho_t$ the correction at leading order. In the perturbative expansion above, we ignore errors exponentially small in $\tau$ that arise from the early stages of the dynamics. We introduce the Drazin inverse of $\lind_t$ to be the operator~\cite{crooks2018drazin}
\begin{equation}
\lind_t^+[X]:=\int_0^\infty\!\D u\,e^{u\lind_t}
\bigl(\omega_t\Tr X-X\bigr).
\label{sm:eq:app-drazin}
\end{equation}
This satisfies $\lind_t^+\lind_t^{\phantom.}[X] = \lind_t^{\phantom.}\lind_t^+[X] =X $ and $\lind_t^+[\omega_t]=0$. Thus, we can rewrite the correction to $\rho_t$ as
\begin{equation}
\delta\rho_t=\lind_t^+[\dot\omega_t] = -\beta \lind_t^+[\J_{\omega_t}[\dot H_t]],
\label{sm:eq:app-lag}
\end{equation}
where we introduced the centred Kubo--Mori map~\cite{petz1993bogoliubov,scandi2019thermodynamiclength}
\begin{equation}
\J_\omega[X]\coloneqq\int_0^1\!\D u\,\omega^{1-u}
\bigl(X-\Tr[\omega X]\eye\bigr)\omega^u,
\label{sm:eq:app-kubo-mori}
\end{equation}
which is obtained from the derivative of the Gibbs state, $\dot\omega_t=-\beta\J_{\omega_t}[\dot H_t]$.
This is self-adjoint with respect to the Hilbert--Schmidt inner product and invertible. Then, the dissipated work defined in the main text (see Eq.~(2)) becomes~\cite{scandi2019thermodynamiclength}
\begin{align}
W_{\rm diss}
&=\int_0^\tau\!\D t\,\Tr\!\left[\dot H_t\lind_t^+[\dot\omega_t]\right]
+O(\tau^{-2}) =-\beta\int_0^\tau\!\D t\,\Tr\!\left[\dot H_t\lind_t^+
[\J_{\omega_t}[\dot H_t]]\right]+O(\tau^{-2})\nonumber\\
&=-\frac1\beta\int_0^\tau\!\D t\,\Tr\!\left[
\dot\omega_t\J_{\omega_t}^{-1}[\lind_t^+[\dot\omega_t]]\right]
+O(\tau^{-2}),
\label{sm:eq:app-work-forms}
\end{align}
where in the second line we equivalently expressed the dissipation in terms of the derivative of the Gibbs state.
For $H_t=\lambda^a(t)X_a$, with fixed control operators and summation over
repeated indices, this yields the work-normalized metric (see Eqs.~(6) and~(7) in the main text)
\begin{equation}
m_{ab}:=-\frac\beta2\left\{
\Tr[X_a\lind_t^+[\J_{\omega_t}[X_b]]]
+\Tr[X_b\lind_t^+[\J_{\omega_t}[X_a]]]\right\},
\qquad
\Sigma [ \{\boldsymbol{\lambda}(t)\}_{t}]
    = \tau \int_0^\tau \D t \, m_{ab}(\boldsymbol{\lambda}) \dot{\lambda}^a(t) \dot{\lambda}^b(t).
\label{sm:eq:app-control-metric}
\end{equation}
Consequently $W_{\rm diss}=\Sigma [ \{\boldsymbol{\lambda}(t)\}_{t}]/\tau+O(\tau^{-2})$. Notice that $m$ is a  metric~\cite{scandi2019thermodynamiclength}. Then, setting $v(t)=\sqrt{m_{ab}\dot\lambda^a\dot\lambda^b}$, the Cauchy--Schwarz inequality gives
\begin{equation}
\Sigma [ \{\boldsymbol{\lambda}(t)\}_{t}]=\tau\int_0^\tau v(t)^2\,\D t
\geq\left(\int_0^\tau v(t)\,\D t\right)^2
=:\ell ^2 [\{ \boldsymbol{\lambda}(t) \}_t]
    \geq L^2(\beta \rightarrow \beta'),
\label{sm:eq:app-length-bound}
\end{equation}
where $\ell[\{ \boldsymbol{\lambda}(t) \}_t]$ is the length of the chosen path, while $L(\beta \rightarrow \beta')$ is the thermodynamic length defined in Eq.~(9) from the main text. The first inequality is saturated at constant thermodynamic speed; the
second by a shortest path. Hence, when a
minimizing path and its constant-speed schedule belong to the set of controllable parameters, the minimum leading-order dissipation, as already defined in Eq.~(3) of the main text, is $\Sigma_{\rm min}(\beta\rightarrow\beta')= L^2(\beta\rightarrow\beta')$.
Otherwise, the distance remains a lower bound. This saturation assumption
is used in the distance-based criteria below. For completeness, we hereby restate the following basic proposition.

\begin{proposition}
    \label{sm:prop:geom_effect}
    The \mname effect occurs if there exist $\beta_2>\beta_1>\beta$ such that
    \begin{align}
        {L}(\beta\rightarrow\beta_1)
        &>{L}(\beta\rightarrow\beta_2).
    \end{align}
\end{proposition}
\begin{proof}
Under the saturation assumption above, squaring the strict inequality
between the nonnegative distances gives
$\Sigma_{\min}(\beta\to\beta_2)<\Sigma_{\min}(\beta\to\beta_1)$,
which is the defining NEAT inequality.
\end{proof}

\section{Minimum system dimension and reduction to incoherent protocols}
In this section, we provide details on the consequences of Proposition 3 in the main text, and in particular on its main corollary:
\begin{corollary}[Minimum system dimension]
    \label{sm:cor:qubitsNoGoapp}
    Classical systems featuring the \mname effect have dimension at least $d\geq 3$ (no \mname
    effect for bits).
    The same holds for quantum systems, under the additional assumption of  detailed-balance
    dynamics and uniqueness of the thermodynamic length geodesic.
\end{corollary}
\begin{proof}
    For classical systems this is a direct consequence of the fact that physical states correspond to normalized probabilities, and as such belong to a $(d-1)$-dimensional manifold. Then, $n\leq d-1$ and the claim follows from Proposition 3 directly. In the quantum case, the manifold of states is $d^2-1$ dimensional. However, for dynamics satisfying detailed balance, Eq.~\eqref{sm:eq:app-length-bound} is locally minimized by purely diagonal drivings (see~\cite{abiuso2020geometric}). Moreover, assuming uniqueness of the geodesics, local minimizers are also global. Then, restricting oneself to optimal drivings reduces the problem to its classical equivalent. For a qubit, this reduces the number of degrees of freedom from $3$ to $1$.
\end{proof}

There are a couple of remarks to be made regarding Cor.~\ref{sm:cor:qubitsNoGoapp}. First, for simplicity, we preferred to restrict to the detailed-balance case in the main text, but we can actually relax this requirement to that of energy covariance of the generator, that is:
\begin{align}
    \lind_s(e^{-i H_s t} Xe^{i H_s t}) =e^{-i H_s t} \lind_s(X)e^{i H_s t}.
\end{align}
This condition implies that the population dynamics completely decouples from the coherence dynamics, which is sufficient for the local minimizer to be diagonal. Indeed, for $\dot H_t = \dot H_t^{\rm diag}+\dot H_t^{\rm coh}$, the dissipation in Eq.~\eqref{sm:eq:app-work-forms} can be rewritten as~\cite{abiuso2020geometric}:
\begin{align}
    W_{\rm diss}
&=-\beta\int_0^\tau\!\D t\,\Tr\!\left[\dot H_t\lind_t^+
[\J_{\omega_t}[\dot H_t]]\right] = -\beta\int_0^\tau\!\D t\,\left(\Tr\!\left[\dot H_t^{\rm diag}\lind_t^+
[\J_{\omega_t}[\dot H_t^{\rm diag}]]\right]+\Tr\!\left[\dot H_t^{\rm coh}\lind_t^+
[\J_{\omega_t}[\dot H_t^{\rm coh}]]\right]\right)
\end{align}
up to corrections of order $O(\tau^{-2})$. Since both traces give a positive contribution, setting one of them to zero locally lowers the value of the dissipated work. This is possible when the endpoints of the trajectory commute, which is exactly the case we are studying.

Secondly, it is worthwhile to provide a sufficient condition for the uniqueness of the thermodynamic length geodesic. This is given by $H$-covariance of the driving dynamics~\cite{cavina2018variational}, that is the requirement that the rates of the Lindbladians are independent of the basis in which the thermalization is taking place. In formulae, suppose that $\lind_s$ thermalizes to $\omega_s$ and $\widetilde{\lind}_s$ thermalizes to $\tilde{\omega}_s= U_s \omega_s U_s^\dagger$. Then, $H$-covariance means that the two Lindbladians are connected by the relation:
\begin{align}
     \lind_s(X) = U_s^\dagger\widetilde{\lind}_s(U_s X U_s^\dagger)U_s.\label{sm:eq:18app}
\end{align}
This additional restriction implies that creating coherences is always detrimental for trajectories with commuting endpoints, since the rates are independent of the basis. Without this additional requirement, one could design a protocol that dissipates extremely fast in one specific basis, and extremely slowly in all other orthogonal directions. Then, if the endpoints were diagonal in one of the slow bases, rotating into the fast basis, changing the spectrum there, and finally rotating back could end up being advantageous. Then, Eq.~\eqref{sm:eq:18app} rules out exactly this kind of protocol.

\section{Geometric conditions for NEAT}
\label{sm:app:geometric-conditions}
In this section, we provide details on the differential-geometric analysis characterizing NEAT and the derivation of Proposition 4 in the main text.

\begin{lem}[Scalar condition for \mname]
\label{sm:lem:neat-scalar}
Suppose $L(\beta\rightarrow\beta')$ is continuous on a cooling interval and differentiable in its interior. There exist $\beta<\beta_1<\beta_2$ in this interval such that $L(\beta\rightarrow\beta_2) <{L}(\beta\rightarrow\beta_1)$ if and only if \begin{equation}
    \left.\frac{\partial}{\partial\beta'}
    L(\beta\rightarrow\beta')\right|_{\beta'=\beta_*}<0
\end{equation}
at some intermediate target $\beta_*$.
\end{lem}
\begin{proof}
For a strict decrease between $\beta_1$ and $\beta_2$, the mean value theorem gives an interior $\beta_*$ whose derivative equals the negative secant slope. Conversely, a negative derivative at $\beta_*$ implies $ L(\beta\to\beta_*+h)< L(\beta\to\beta_*)$ for sufficiently small $h>0$. Taking $\beta_1=\beta_*$ and
$\beta_2=\beta_*+h$ proves the converse.
\end{proof}

\textbf{Proposition 4} (Geometric condition for NEAT)
\label{sm:prop:neat-geometric}
\textit{Assume that the minimizing geodesic $\varsigma$ from the initial state to every intermediate target is unique. Then \mname occurs if and only if, at some intermediate target $\beta_*$, the scalar product $(\cdot,\cdot)_m$ in metric $m$ between the tangent $\mathbf{v}_\varsigma(\beta_*)$ of the geodesic ending at $\omega(\beta_*)$ and the tangent $\mathbf{v}_\omega(\beta_*)$ of the Gibbs curve is negative, i.e.
\begin{equation}
    \exists\ \beta_*\ \text{ s.t. }\
    \big( \mathbf{v}_\varsigma(\beta_*),\mathbf{v}_\omega(\beta_*) \big)_{m} < 0.
    \label{sm:eq:neat-endpoint-pairing}
\end{equation}
}

\begin{proof}
The main property used here is the smoothness of the geodesic distance. Given a starting point, its \textit{cut locus} is the set of endpoints at which the minimizing geodesic is not uniquely defined, or the endpoint is conjugate to the starting point. Away from the cut locus (and the initial point itself) the minimizing geodesic is unique and the geodesic distance is smooth.

Fix $\beta$. For each target $\beta'$ away from the cut locus, let $\varsigma(\beta')$ be the minimizing geodesic from $\omega(\beta)$ to $\omega(\beta')$, parametrized proportionally to arc length with $s\in[0,1]$. Let $L(\beta\to\beta')=d_m\big(\omega(\beta),\omega(\beta')\big)$ be the geodesic distance between $\omega(\beta)$ and $\omega(\beta')$. By the property above, $\nabla\, d_m(\omega(\beta),\cdot)\big|_{\omega(\beta')} = \mathbf{v}_\varsigma(\beta')$, where $\mathbf{v}_\varsigma(\beta')$ is the tangent at $\omega(\beta')$ of the minimizing geodesic arriving there. Since the Gibbs curve is $C^1$, using the chain rule,
\begin{equation}
\frac{\partial L(\beta\to\beta')}{\partial\beta'} = \big(\mathbf{v}_\varsigma(\beta'),\,\mathbf{v}_\omega(\beta')\big)_{m}\Big|_{\omega(\beta')},
\label{sm:eq:app-first-variation}
\end{equation}
with $\mathbf{v}_\omega(\beta')$ the tangent at $\omega(\beta')$ of the Gibbs curve. Thus, to first order the length changes only through the component of the endpoint velocity along the geodesic.
By \eqref{sm:eq:app-first-variation}, $\partial_{\beta'}L(\beta\to\beta')<0$ at $\beta'_{\phantom{*}}=\beta_*^{\phantom{'}}$ if and only if $\big(\mathbf{v}_\varsigma(\beta_*),\mathbf{v}_\omega(\beta_*)\big)_m<0$. Hence the distance from the initial state decreases at some intermediate
target exactly when \eqref{sm:eq:neat-endpoint-pairing} holds, which is the claim.
\end{proof}

\section{NEAT effect in incoherent processes}
\label{sm:app:DB}

\subsection{Dynamics of incoherent processes}
\label{sm:app:classical_DB}
For $H_t=\sum_{a=1}^{d}e_a(t)\ketbra{a}$, assume that the generator preserves the diagonal subspace at all times. An initially diagonal state remains diagonal, and its population vector $\boldsymbol\rho$ obeys
\begin{equation}
\dot{\boldsymbol\rho}_t=M_t\boldsymbol\rho_t,\qquad
\boldsymbol1^{\mathsf T}M_t=0,\qquad
M_t\boldsymbol\omega_t=0,\qquad (M_t)_{ab}\geq0\quad(a\ne b).
\end{equation}
Here $\boldsymbol\omega_t$ is a vector with components $\omega_a(t)=e^{-\beta e_a(t)}/Z_t$, and $Z_t=\sum_{a=1}^d e^{-\beta e_a(t)}$. We also introduce the column vector $\boldsymbol1=(1,\ldots,1)^{\mathsf T}$, where $\mathsf T$ denotes the transpose operation. With the current convention, $(M_t)_{ab}$ can be interpreted as the rate of the transition from level $b$ to level $a$. Moreover, it should be noticed that the Drazin inverse of the Lindblad operator reduces to the Drazin inverse of $M_t$ if we restrict to diagonal states.

On diagonal observables, the Kubo--Mori map reduces to the covariance matrix, so that on any traceless vector $\boldsymbol{v}$ (i.e., such that $\boldsymbol1^{\mathsf T}\boldsymbol v=0$) it acts as
\begin{equation}
(J_{\boldsymbol\omega}\boldsymbol{v})_a = \omega_a v_a -\boldsymbol\omega^{\mathsf T}\boldsymbol{v},
\qquad
(J_{\boldsymbol\omega}^{-1}\boldsymbol v)_a = \omega^{-1}_a v_a+c,
\end{equation}
where $c$ is an arbitrary constant that drops out on traceless vectors. Equation~\eqref{sm:eq:app-work-forms} therefore becomes
\begin{equation}
W_{\rm diss}=-\frac1\beta\int_0^\tau\!\D t\,\sum_{ab}\,(\dot{\omega}_t)_a\,(\omega_t^{-1})_a\,(M_t^+)_{ab}\,(\dot{\omega}_t)_b + O(\tau^{-2}).
\label{sm:eq:app-population-work}
\end{equation}
We assume the population dynamics $M$ to have a unique fixed point, and to satisfy detailed balance. The latter means that its rates satisfy the symmetry condition
$M_{ab}\omega_b=M_{ba}\omega_a$. Then, let us introduce the matrix
\begin{equation}
R_{ab}=-\omega_a^{-1/2}M_{ab}^{\phantom.}\,\omega_b^{1/2}.
\end{equation}
Thanks to the symmetry in the rates, it holds that $ R_{ba} = R_{ab}$. Moreover, the vector $\boldsymbol q=\sqrt{\boldsymbol\omega}$ (where the square root is component-wise) satisfies $R\boldsymbol q=0$. It should also be noticed that the Drazin inverse $R^+$ is simply given by $R^+_{ab}=-\omega^{-1/2}_aM^+_{ab}\,\omega_b^{1/2}$ (since the transformation $M\rightarrow R$ is invertible, and compatible with matrix multiplication). Then, assuming a driving of the form $H_t=\sum_{a}e_a(t)\ketbra{a}{a}$, passing from $\omega_a(t)$ to $q_a(t) := \sqrt{\omega_a(t)}$, we can rewrite the dissipated work as
\begin{equation}
W_{\rm diss}=\frac4\beta\int_0^\tau\!\D t\,\sum_{ab}
\,(\dot{q}_t)_a \,(R_t^+)_{ab}
\,(\dot{q}_t)_b+O(\tau^{-2}).
\label{sm:eq:app-diagonal-metric}
\end{equation}
From this equation we can deduce the expression for the metric in the space of square-root populations.

\subsection{The qutrit example}
\label{sm:app:trit_mpemba}
For the equilibrium populations of the qutrit, we use the spherical coordinates
\begin{equation}
\boldsymbol \omega=
\begin{pmatrix}\sin^2\theta\cos^2\varphi\\\cos^2\theta\\\sin^2\theta\sin^2\varphi\end{pmatrix},
\qquad
\boldsymbol q=
\begin{pmatrix}\sin\theta\cos\varphi\\\cos\theta\\\sin\theta\sin\varphi\end{pmatrix},
\qquad 0<\theta< \pi/2,\qquad 0<\varphi<\pi/2.
\end{equation}
Define the orthonormal tangent vectors
$\boldsymbol e_\theta=\partial_\theta\boldsymbol q$ and
$\boldsymbol e_\varphi=(\sin\theta)^{-1}\partial_\varphi\boldsymbol q$, and choose
\begin{equation}
R=r_\theta\boldsymbol e_\theta\boldsymbol e_\theta^{\mathsf T}
+r_\varphi\boldsymbol e_\varphi\boldsymbol e_\varphi^{\mathsf T}.
\label{sm:eq:app-qutrit-generator}
\end{equation}
The positive quantities $r_\theta,r_\varphi$ are the relaxation rates of the two tangent modes. It can be directly verified that $R\boldsymbol{q} = 0$ (since $\boldsymbol{q}$ is orthogonal both to $\boldsymbol e_\theta$ and $\boldsymbol e_\varphi$), so that the corresponding $M$ fixes $\boldsymbol\omega$.
We can write the metric in these coordinates starting from Eq.~\eqref{sm:eq:app-diagonal-metric}. Indeed, we have that $ \dot {\boldsymbol{q}} = \dot\theta \boldsymbol{e}_\theta + \dot\varphi\sin\theta\boldsymbol{e}_\varphi$, so Eq.~\eqref{sm:eq:app-diagonal-metric} becomes:
\begin{align}
    W_{\rm diss}=\int_0^\tau\!\D t\,(\dot\theta^2 \,m_{\theta\theta} +  \dot\varphi^2 \,m_{\varphi\varphi}  + 2 \dot\theta\dot\varphi \,m_{\theta\varphi}),
\end{align}
where $m_{\theta\theta} := {4}{\beta}^{-1}\sum_{ab} (\partial_\theta\boldsymbol{q})_a R^+_{ab} (\partial_\theta\boldsymbol{q})_b$ (and similarly for the other components).
Then, we can explicitly compute each of the components to obtain:
\begin{equation}
m_{\theta\theta}=\frac4{\beta r_\theta},\qquad
m_{\varphi\varphi}=\frac{4\sin^2\theta}{\beta r_\varphi},\qquad
m_{\theta\varphi}=0.
\label{sm:eq:app-qutrit-metric}
\end{equation}
For constant rates, the geodesic equations are
\begin{equation}
\ddot\theta-\frac{r_\theta}{r_\varphi}\sin\theta\cos\theta\,\dot\varphi^2=0,
\qquad \ddot\varphi+2\cot\theta\,\dot\theta\dot\varphi=0.
\label{sm:eq:app-qutrit-geodesics}
\end{equation}
They conserve $\sin^2\theta\,\dot\varphi$ and the metric speed~\cite{carroll2004spacetime}. The two rates are generally different. To simplify the form of the metric, we introduce, without loss of generality, the rescaled variable $\psi=\sqrt{r_\theta/r_\varphi}\,\varphi$. Transforming $m_{\varphi\varphi}$ of Eq.~\eqref{sm:eq:app-qutrit-metric} gives $m_{\psi\psi}={(4\sin^2\theta)}/{\beta r_\theta}$; together with $m_{\theta\theta}$, this is the round-sphere metric of radius $2/\sqrt{\beta r_\theta}$. This means that in these coordinates, the geodesics are given by great circles, so that one can easily compute the length of optimal paths. Indeed, consider two vectors $\widetilde{\boldsymbol q}_A$ and $\widetilde{\boldsymbol q}_B$, corresponding to expressing two states $\boldsymbol{\omega}_A$ and $\boldsymbol{\omega}_B$ in the $(\theta, \psi)$ coordinates. Since they are unit vectors, their scalar product corresponds to the cosine of the angle $\kappa$ between them, i.e.:
\begin{equation}
\widetilde{\boldsymbol q}_A^{\mathsf T}\widetilde{\boldsymbol q}_B=\cos\kappa
=\cos\theta_A\cos\theta_B+
\sin\theta_A\sin\theta_B\cos(\psi_B-\psi_A).
\end{equation}
Then, the shorter great-circle arc has length $2\kappa/\sqrt{\beta r_\theta}$~\cite{lee2018riemannian}.
Substituting back $\varphi$ yields
\begin{equation}
 L(\boldsymbol\omega_A,\boldsymbol\omega_B)
=\frac2{\sqrt{\beta r_\theta}}\arccos\!\left[
\cos\theta_A\cos\theta_B+\sin\theta_A\sin\theta_B
\cos\!\left(\sqrt{\frac{r_\theta}{r_\varphi}}(\varphi_B-\varphi_A)\right)\right],
\label{sm:eq:app-qutrit-distance}
\end{equation}
corresponding to the geodesics~\cite{lee2018riemannian}
\begin{equation}
\widetilde{\boldsymbol q}(t)=
\frac{\sin[(1-t/\tau)\kappa]}{\sin\kappa}\widetilde{\boldsymbol q}_A
+\frac{\sin[(t/\tau)\kappa]}{\sin\kappa}\widetilde{\boldsymbol q}_B.
\end{equation}
The corresponding energy driving can be recovered as follows. First, notice that $\theta(t)=\arccos\widetilde q_2(t)$ and
$\varphi(t)=\sqrt{r_\varphi/r_\theta}\,
\operatorname{atan}(\widetilde q_3(t) / \widetilde q_1(t))$;
then, we can invert the relation between the coordinates of thermal states and their Hamiltonian, to obtain that the energy at each time should satisfy:
\begin{equation}
e_1(t)=0,\qquad
e_2(t)=-\beta^{-1}\log\frac{\cot^2\theta(t)}{\cos^2\varphi(t)},\qquad
e_3(t)=-\beta^{-1}\log\tan^2\varphi(t).
\end{equation}
This gives the explicit driving needed to minimize the dissipation.

\subsection{Birth-death master equations}
We use levels $a=1,\ldots,d$ with $e_1=0<e_2<\cdots<e_d$.
Bond $a=1,\ldots,d-1$ connects levels $a,a+1$ and has gap
$\Delta_a=e_{a+1}-e_a$. Thus $A_a=\ket a\bra{a+1}$ and
\begin{equation}
\lind_t^{\rm diss}[\rho]=\sum_{a=1}^{d-1}
\left(\gamma_a^\downarrow\mathcal D_{A_a}[\rho]
+\gamma_a^\uparrow\mathcal D_{A_a^\dagger}[\rho]\right),\qquad
\mathcal D_A[\rho]=A\rho A^\dagger-\tfrac12\{A^\dagger A,\rho\}.
\end{equation}
Here $\boldsymbol\rho=(p_1,\ldots,p_d)^{\mathsf T}$ is a population vector, and $\gamma_a^\uparrow,\gamma_a^\downarrow$ are the rates for $a\to a+1$ and $a+1\to a$, respectively.
The population equation and detailed balance are
\begin{align}
(M\boldsymbol\rho)_a&=\gamma_{a-1}^\uparrow p_{a-1}
+\gamma_a^\downarrow p_{a+1}
-(\gamma_{a-1}^\downarrow+\gamma_a^\uparrow)p_a,\\
\gamma_a^\uparrow\omega_a&=\gamma_a^\downarrow\omega_{a+1},
\qquad \gamma_a^\uparrow/\gamma_a^\downarrow=e^{-\beta\Delta_a}.
\end{align}
Boundary terms are omitted, with $\gamma_0^{\uparrow,\downarrow}=\gamma_d^{\uparrow,\downarrow}=0$. Following the main-text cumulative definition,
\begin{equation}
\Omega_a=\sum_{b=1}^a\omega_b,\quad a=1,\ldots,d-1,\qquad
w_a=\gamma_a^\uparrow\omega_a
=\gamma_a^\downarrow\omega_{a+1}>0.
\label{sm:eq:app-bd-conductance}
\end{equation}
Thus $w_a$ is the conductance across the boundary of $\Omega_a$;
$\Omega_d=1$ is not an independent coordinate.
For the bosonic bath, with $\hbar=1$,
\begin{equation}
\gamma_a^\downarrow=|c_a|^2\mathcal J(\Delta_a)[n_\beta(\Delta_a)+1],\qquad
\gamma_a^\uparrow=|c_a|^2\mathcal J(\Delta_a)n_\beta(\Delta_a),\qquad
n_\beta(\nu)=(e^{\beta\nu}-1)^{-1},
\end{equation}
with $c_a$ being the system coupling matrix element on bond $a$, $\mathcal J(\nu)$ the bath spectral density, and $n_\beta(\nu)$ the Bose occupation at frequency $\nu$.
A common small coupling prefactor can be restored to all rates to ensure
weak coupling. It rescales every distance by the same inverse square root
and leaves the cooling ordering unchanged.

\subsubsection{Drazin inverse of the rate matrix}
Let $\boldsymbol v$ have zero sum and solve $M\boldsymbol h=\boldsymbol v$
with $\sum_a h_a=0$. Define $V_a=\sum_{b=1}^a v_b$ and
$j_a=\gamma_a^\uparrow h_a-\gamma_a^\downarrow h_{a+1}$,
with $j_0=j_d=0$. The continuity equation gives
$v_a=j_{a-1}-j_a$, hence $j_a=-V_a$. With $f_a=h_a/\omega_a$,
detailed balance yields $f_{a+1}-f_a=V_a/w_a$. Therefore
\begin{equation}
(M^+\boldsymbol v)_a=\omega_a\left[\Phi_a-
\sum_{b=1}^{d}\omega_b\Phi_b\right],\qquad
\Phi_1=0,\qquad \Phi_a=\sum_{b=1}^{a-1}\frac{V_b}{w_b}.
\label{sm:eq:app-bd-drazin}
\end{equation}
Together with $M^+\boldsymbol\omega=0$, this specifies the inverse on the
full population space~\cite{meyer1975groupinverse}.

\subsubsection{Thermodynamic metric}
Substitution into Eq.~\eqref{sm:eq:app-population-work} gives the quadratic form
\begin{equation}
-\frac1\beta\sum_{a=1}^{d}\frac{v_a(M^+\boldsymbol v)_a}{\omega_a}
=-\frac1\beta\sum_{a=1}^{d}v_a\Phi_a
=\frac1\beta\sum_{a=1}^{d-1}\frac{V_a^2}{w_a}.
\end{equation}
Taking $\boldsymbol v=\dot{\boldsymbol\omega}$ gives $V_a=\dot\Omega_a$. Thus the metric in cumulative coordinates is
\begin{equation}
m_{ab}(\boldsymbol\Omega)=\frac{\delta_{ab}}{\beta w_a(\boldsymbol\Omega)},
\qquad a,b=1,\ldots,d-1,
\label{sm:eq:app-bd-cdf-metric}
\end{equation}
Here $\delta_{ab}$ is the Kronecker delta, equal to one for $a=b$ and zero otherwise. The populations satisfy
\begin{equation}
\omega_a=\Omega_a-\Omega_{a-1},\quad
\Omega_0=0,\quad\Omega_d=1,\qquad
w_a=|c_a|^2\mathcal J(\Delta_a)
\frac{\omega_a\omega_{a+1}}{\omega_a-\omega_{a+1}}.
\label{sm:eq:app-bd-weight-explicit}
\end{equation}
The domain satisfies $\omega_1>\cdots>\omega_d>0$ and normalization.
Let us denote by $\mathds1_{a\leq c}$ the indicator, equal to one when $a\leq c$ and zero otherwise.
For completeness, $\partial\Omega_a/\partial e_b
=-\beta\omega_b(\mathds1_{b\leq a}-\Omega_a)$ gives the metric
in the energy eigenbasis,
\begin{equation}
m_{e_a e_b}=\beta\omega_a\omega_b\sum_{c=1}^{d-1}
\frac{(\mathds1_{a\leq c}-\Omega_c)(\mathds1_{b\leq c}-\Omega_c)}{w_c}.
\label{sm:eq:app-bd-energy-metric}
\end{equation}
Before fixing $e_1$, its only null direction is a uniform energy shift.

\subsubsection{Static bond conductances}

We now show that within this model, even with some level of anisotropy, if such an anisotropy is constant and cannot be modulated with the energy levels (or can but isn't) no NEAT effect can occur.

\begin{lem}[No same-branch NEAT for static bond conductances]
\label{sm:lem:static-conductances}
For a birth--death process with constant positive bond conductances
$w_a=\widehat w_a$, suppose the full ordered-population domain and
constant-speed schedules are accessible. Along the positive-temperature
Gibbs family of a fixed ordered Hamiltonian, the thermodynamic distance
from the initial state is nondecreasing as the target moves away from it
on either thermal branch. Hence no same-branch NEAT effect occurs.
\end{lem}
\begin{proof}
The two thermal branches are cooling ($\beta'>\beta$) and heating ($0<\beta'<\beta$), with the initial $\beta$ fixed.
If $w_a=\widehat w_a>0$ is constant, the cumulative metric is Euclidean
after a diagonal rescaling. With access to the full ordered-population
domain, the straight segment between Gibbs endpoints is admissible, so
\begin{equation}
L^2(\beta\to\beta')=
\frac1\beta\sum_{a=1}^{d-1}
\frac{[\Omega_a(\beta')-\Omega_a(\beta)]^2}{\widehat w_a}.
\end{equation}
Along the Gibbs family of an ordered, nonconstant $H_0$, with $\beta'>0$,
\begin{equation}
\frac{\D\Omega_a}{\D\beta'}=
\Omega_a(1-\Omega_a)
\bigl(\langle e\rangle_{b>a}-\langle e\rangle_{b\leq a}\bigr)>0.
\end{equation}
The conditional averages are $\langle e\rangle_{b\leq a}=\sum_{b=1}^a e_b\omega_b/\Omega_a$ and $\langle e\rangle_{b>a}=\sum_{b=a+1}^d e_b\omega_b/(1-\Omega_a)$, with all populations evaluated at $\beta'$. Every squared
difference increases away from $\beta$ on either thermal branch, excluding
same-branch NEAT. Constant conductances do not mean constant upward and
downward rates: their ratio still obeys detailed balance.
\end{proof}

\section{Consequences for thermal machines}
\label{sm:app:engines}
In this section, we provide a discussion on the performance improvements that anomalous effects yield on thermal machines. The starting observation is simple: a basic Otto cycle alternates 2 adiabatic strokes with 2 isochoric strokes, whereas Carnot cycles alternate 2 adiabatic strokes with 2 isothermal strokes~\cite{quan2007quantum,seifert2012stochastic}. In both cases, the adiabatic strokes are typically taken to be instantaneous and dissipationless.
It then follows that the finite-time performance of these thermal machines highly depends on the non-adiabatic strokes.

In particular, the Mpemba effect can be seen as the speedup of certain fixed-Hamiltonian relaxations --isochoric strokes, in the language of thermal engines-- while their amplitude increases. As such, the Mpemba effect naturally boosts the performance of Otto engines, as shown, e.g., in Refs.~\cite{lin2022power,liu2024speeding}.

The NEAT effect can be seen as being the analogous feature of close-to-equilibrium isothermal strokes.  NEAT therefore can be used to improve the performance of \emph{finite-time Carnot engines}, as we analyze in the following.

\subsection{Carnot-cycle geometry}

For the reference Carnot cycle in Fig.~\ref{sm:fig:carnot-cycle}, use natural units with $k_{\rm B}=1$ and take
$T_h>T_c$, $\beta_{h,c}=T_{h,c}^{-1}$ and $0<\lambda<1$.
The equilibrium corners are defined as the state-Hamiltonian pairs
\begin{align}
A&=(\omega_{\rm lo},H),&
B&=(\omega_{\rm hi},\lambda H),&
C&=(\omega_{\rm hi},\lambda(T_c/T_h)H),&
D&=(\omega_{\rm lo},(T_c/T_h)H),
\end{align}
where $\omega_{\rm lo}=\omega_H(\beta_h)$ and $\omega_{\rm hi}=\omega_{\lambda H}(\beta_h)$.
We define $\Delta S=S(\omega_{\rm hi})-S(\omega_{\rm lo})>0$, with $S(\rho)=-\Tr\rho\ln\rho$ the (von Neumann) entropy of the system.

\paragraph{Heat balances and the meaning of dissipated work.}
The work performed on the system and the heat absorbed by the system take the usual form
\begin{align}
    W = \int_0^\tau \D t\, \tr [\rho_t \dot{H}_t]; \quad Q=\int_0^\tau  \D t\, \tr [\dot{\rho}_t {H}_t].
\end{align}
The cycle is oriented in order to let the heat $Q_h$ be positive when absorbed from the hot bath (cf. Fig.~\ref{sm:fig:carnot-cycle}).
\begin{figure}
    \centering
    \includegraphics[width=.55\textwidth]{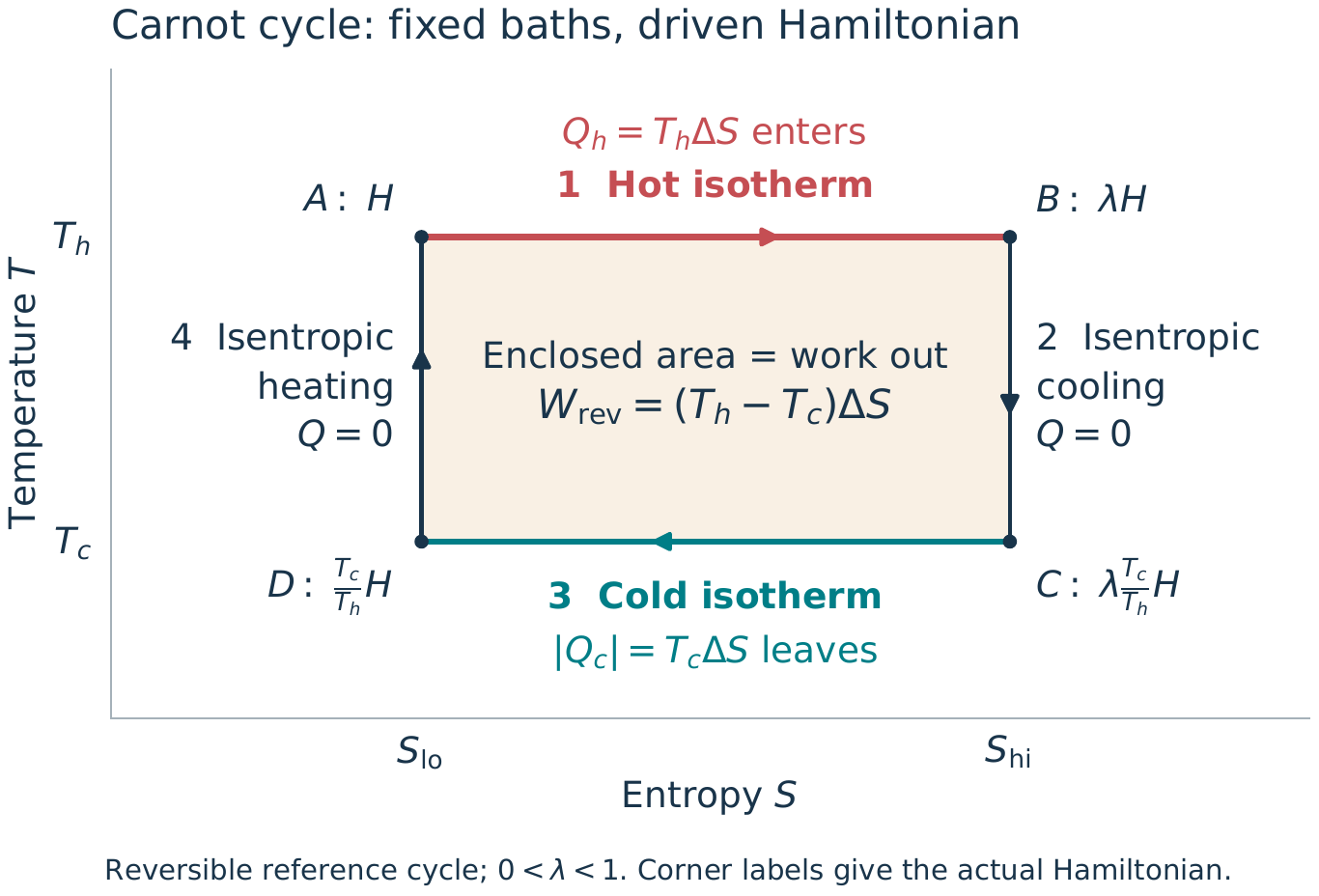}
    \caption{Reversible Carnot cycle ($k_{\rm B}=1$), traversed clockwise.
    Corner labels give the Hamiltonians, with $0<\lambda<1$.
    Isotherms exchange heat; isentropes preserve populations under Hamiltonian
    rescaling. The enclosed area is the reversible work. Finite-time
    isothermal losses reduce it by $ L_h^2/\tau_h+ L_c^2/\tau_c$.}
    \label{sm:fig:carnot-cycle}
\end{figure}
As customary, we take the isentropic strokes to have negligible durations, and we define $\tau_h,\tau_c$ to be the isothermal durations of the $A\rightarrow B$ and $C\rightarrow D$ transformations, respectively. As such, the total duration of the cycle is
\begin{align}
    \tau=\tau_h+\tau_c,
\end{align}
and heat is exchanged solely during the hot and cold isotherms. As a consequence, the total output of the cycle is given by
\begin{align}
    W_{\rm out}=Q_h+Q_c.
\end{align}
At leading slow-driving order we have, in the notation of the main text, the latter equal to (cf. \cite{cavina2017slowdynamics,scandi2019thermodynamiclength,abiuso2020geometric,abiuso2020optimal})
\begin{align}
Q_h&=T_h\Delta S-\frac{\Sigma_h}{\tau_h},&
Q_c&=-T_c\Delta S-\frac{\Sigma_c}{\tau_c},& W_{\rm out}&=(T_h-T_c)\Delta S  -\frac{\Sigma_h}{\tau_h}-\frac{\Sigma_c}{\tau_c}.
\end{align}
The finite-time protocol-dependent dissipations $\Sigma_j/\tau_j$ are positive as implied by the second law, and therefore always detrimental to the engine's output. It follows that for general performance optimization~\cite{abiuso2020geometric} both $\Sigma_h$ and $\Sigma_c$ should be minimized over the possible control protocols. Having this in mind, we substitute them with their attainable geodesic minima  $\Sigma\geq L^2 $ (cf. main text) to obtain
\begin{align}
Q_h&=T_h\Delta S-\frac{ L_h^2}{\tau_h},&
Q_c&=-T_c\Delta S-\frac{L_c^2}{\tau_c},\\
W_{\rm out}&=(T_h-T_c)\Delta S
-\frac{ L_h^2}{\tau_h}-\frac{ L_c^2}{\tau_c},&
\eta&=\frac{W_{\rm out}}{Q_h}=1-\frac{T_c}{T_h}\left(\frac{1+\frac{L_c^2}{\tau_c T_c\Delta S}}{1-\frac{L_h^2}{\tau_h T_h\Delta S}}\right).
\label{sm:eq:app-engine-work}
\end{align}
These expressions characterize the engine's performance in the slow-driving or low-dissipation limit.

\paragraph{A simple observation.} It should be clear already at this point that NEAT is a good candidate for the improvement of performance of a finite-time Carnot engine. All the figures of merit based on work output $W_{\rm out}$, efficiency $\eta$, or power $W_{\rm out}/(\tau_h+\tau_c)$ are decreasing functions of $\{L_h,L_c\}$, while they increase in $\Delta S$. NEAT precisely increases the latter, while decreasing the former.
In the following we make this observation more evident by optimizing the strokes' durations $\tau_{h,c}$ and later assuming scaling-symmetry relations between the hot and cold baths.

\subsection{Optimizing the stroke durations}

Writing $W_{\rm rev}=(T_h-T_c)\Delta S$ and
$\tau=\tau_h+\tau_c$, optimal time allocation (that is, maximum $W_{\rm out}$ at fixed $\tau=\tau_h+\tau_c$) yields
\begin{align}
    \frac{\tau_h}{\tau_c}=\frac{ L_h}{L_c} \quad \text{and} \quad
P(\tau)=\frac{W_{\rm rev}}{\tau}- \frac{(L_h+L_c)^2}{\tau^2}=\frac{(T_h-T_c) \Delta S}{\tau}- \frac{(L_h+L_c)^2}{\tau^2}.
\end{align}
The maximum occurs at $\tau^*=2 ( L_h+ L_c)^2/W_{\rm rev}$, yielding in this regime
\begin{equation}
P_{\rm max}=\frac{[ (T_h-T_c)\Delta S]^2}{4 ( L_h+ L_c)^2}.
\label{sm:eq:app-engine-pmax-general}
\end{equation}

Once again, the above expressions increase monotonically in the entropic swing $\Delta S$, and decrease in the total hot-plus-cold thermodynamic length $L_h+L_c$.
Notice that a NEAT effect instance can in principle be associated with one of the two baths only (namely increasing $\Delta S$ while decreasing $L_j$ in one of the two isotherms $j=h,c$). The behaviour of the ratio $\Delta S/(L_h+L_c)$ is not mathematically guaranteed in such a case, whereas a joint increase of $\Delta S$ and decrease of $L_h+L_c$ necessitates NEAT at least for one of the two isotherms.
Finally, if the two baths satisfy a scaling symmetry of the form
\begin{align}
    L_h=\gamma L_c
\end{align}
with $\gamma$ being independent of the optimized control trajectory, both baths yield NEAT effects on the same endpoints.
An example of such scaling is provided by the Ohmic scaling that follows.

\subsection{Ohmic scaling and the NEAT advantage on maximum power}
In the main text we assumed for the sake of simplicity that the hot and cold metrics satisfy~\cite{abiuso2020geometric,abiuso2020optimal}
\begin{equation}
\frac{L_c}{L_h}=(T_c/T_h)^{-\alpha/2}.
\end{equation}
Here, the exponent $\alpha$ can be derived assuming baths with the same power-law spectral density, $\mathcal J(\nu)\propto\nu^\alpha$ (Ohmic for $\alpha=1$)~\cite{cavina2017slowdynamics,abiuso2020geometric}.
Substitution gives
\begin{equation}
P_{\max}=\frac{\Delta S^2}{L_h^2}
\frac{ (T_h-T_c)^2}
{4T_h[1+(T_c/T_h)^{(1-\alpha)/2}]^2}.
\end{equation}
This scaling applies to suitable power-law baths, whereas it is not automatic for
a structured spectrum (under the proportional-metric assumption, both distances scale together; more generally, an improved maximum power requires an increase of $\Delta S/(L_h+L_c)$).
None of these changes alters the Carnot bound on maximum efficiency
$1-T_c/T_h$, which is recovered for $\tau_{h,c}\rightarrow\infty$.

\putunitbibliography
\end{bibunit}
\end{document}